\documentclass[12pt,oneside,english]{amsart}
\usepackage[utf8]{inputenc}
\usepackage[a4paper]{geometry}
\usepackage{footnote}
\usepackage{color}
\usepackage{babel}
\usepackage{array}
\usepackage{float}
\usepackage{url}
\usepackage[foot]{amsaddr}

\usepackage{enumitem}
\usepackage{multirow}
\usepackage{amsmath}
\usepackage{amsthm}
\usepackage{graphicx}
\usepackage{setspace}
\usepackage[authoryear]{natbib}
\usepackage[unicode=true,pdfusetitle,
 bookmarks=true,bookmarksnumbered=false,bookmarksopen=false,
 breaklinks=false,pdfborder={0 0 1},backref=section,colorlinks=true]
 {hyperref}
\hypersetup{
 citecolor=cite-blue}

\makeatletter

  \theoremstyle{plain}
  \newtheorem{prop}{\protect\propositionname}

   \newtheorem{fact}{Fact}
  \theoremstyle{plain}
  \newtheorem{cor}{\protect\corollaryname}
  \theoremstyle{remark}
  \newtheorem{rem}{\protect\remarkname}
  \theoremstyle{definition}
  \newtheorem{defn}{\protect\definitionname}
\theoremstyle{plain}
\newtheorem{thm}{\protect\theoremname}
  \theoremstyle{plain}
  
\newtheorem{example}{Example}
\definecolor{cite-blue}{RGB}{0,0,204}
\date{}

\makeatother

  \providecommand{\definitionname}{Definition}
  \providecommand{\lemmaname}{Lemma}
  \providecommand{\propositionname}{Proposition}
  \providecommand{\remarkname}{Remark}
\providecommand{\corollaryname}{Corollary}
\providecommand{\theoremname}{Theorem}

\begin{document}
%\doublespace
\onehalfspacing

\title{Strategy-proof Popular Mechanisms}
\author{Mustafa O\u{g}uz Afacan}
\author{In\'{a}cio B\'{o}}

\address{\textbf{Afacan}: Sabanc\i{} University, Faculty of Art and Social Sciences, Orhanli,
34956, Istanbul, Turkey. e-mail:  Email:\href{mailto:mafacan@sabanciuniv.edu}{mafacan@sabanciuniv.edu}.}

\address{\textbf{Bó}: Corresponding Author. China Center for Behavioral Economics and Finance, Southwestern University of Finance and Economics, Chengdu, China; website: \protect\url{http://www.inaciobo.com}; email:
\href{mailto:inaciog@gmail.com}{inaciog@gmail.com}.  .}

\date{}
%\thanks{The authors thank Samson Alva, Lars Ehlers, Rustamdjan Hakimov, Aleksei Kondratev, Alexander Nesterov, and participants of the 2019 Lisbon Game Theory Meetings for helpful comments. All errors are our own.}
\maketitle
\begin{abstract}
We consider the allocation of indivisible objects when agents have preferences over their own allocations, but share the ownership of the resources to be distributed. Examples might include seats in public schools, faculty offices, and time slots in public tennis courts. Given an allocation, groups of agents who would prefer an alternative allocation might challenge it. An assignment is popular if it is not challenged by another one. By assuming that agents' ability to challenge allocations can be represented by weighted votes, we characterize the conditions under which popular allocations might exist and when these can be implemented via strategy-proof mechanisms. Serial dictatorships that use orderings consistent with the agents' weights are not only strategy-proof and Pareto efficient, but also popular, whenever these assignments exist. We also provide a new characterization for serial dictatorships as the only mechanisms that are popular, strategy-proof, non-wasteful, and satisfy a consistency condition.
\end{abstract}

\section{Introduction}
\label{sec:Introduction}

Consider the problem of assigning a finite set of indivisible objects to a finite set of agents who share a loosely defined ``common ownership'' of the objects, but only care about the object that is assigned to themselves. Examples of problems like these include the assignment of offices to faculty in a department, time slots in public tennis courts, or seats in public schools.\footnote{While it might be argued that in these problems individuals have preferences not only about their own allocation but also about their peers in schools or office neighbors, we take the standard approach of assuming that individuals are indifferent between any two allocations in which their assignment does not change.} In these settings, standard properties such as Pareto efficiency or the core might not correctly capture the nature of the conflict created by these preferences combined with the shared ownership of the objects.

Take Pareto efficiency. An allocation of objects to agents is Pareto efficient if there is no alternative allocation in which no agent is worse off and at least one is strictly better off.  When this is the objective of the designer, there are many mechanisms with good outcome and incentive properties available \citep{Shapley1974-ir,Abdulkadiroglu2003-nx,Pycia2017-xe}. Many allocations, however, are Pareto efficient because to improve the allocation of many agents it would be necessary to make one or a few agents worse off.\footnote{Consider, for example, a problem consisting of an arbitrarily large number of agents, $n$, and the same number of unit-copy objects. Let $i_1,..,i_n$ and $a_1,..,a_n$ be enumerations of the agents and the objects, respectively. The earlier an object comes, the more preferred it is by all the agents. Consider an allocation $\mu$ where each agent $i_k$ is given the object $a_k$. This would be, for example, the outcome of a serial dictatorship in which agents are ordered by their indexes. This allocation is Pareto efficient. Notice, however, that under an alternative matching $\mu'$ in which $i_1$ is assigned $a_n$, and each agent $i_k$ to $a_{k-1}$, all the agents but $i_1$ would be strictly better off.} When all the agents share the ownership of the objects, however, such allocations could be challenged by those who would be better-off under an alternative allocation. The concept of \textit{popularity} considers these issues. A matching $\mu$ is \textbf{more popular} than matching $\mu'$ if the number of agents who strictly prefer their assignments under $\mu$ to those under $\mu'$ is larger than the number of agents who strictly prefer their assignments under $\mu'$ to those  under $\mu$. Here we would say that $\mu'$ is challenged by $\mu$. A matching is \textbf{popular} if there is no matching that is more popular than it. Notice that, similarly to the test of whether an allocation is in the core, we test whether coalitions of agents could reallocate the objects that they own in a way that makes them better off. Differently from the core, however, agents are not endowed with part of the objects to be allocated. Moreover, we consider not only whether the agents in the coalition would be better off, but also how many outside of that coalition would be worse off. In our setup, therefore, no subset of agents has an endowment of some set of objects to reallocate between themselves, and every agent has, to some extent, a say on the allocation of all objects.

When considering the standard notion of popularity, we hit strong negative results. Popular matchings might not exist (Corollary \ref{cor:popMatchingMayNotExist})---a previously known fact---and there is no mechanism that is strategy-proof and produces a popular matching when one exists, regardless of the allocations that it produces when those do not exist (Corollary \ref{cor:noSPPopularMech}). In other words, even if we ignore the fact that popular matchings might not exist, there is no ``good'' method of implementing them when they do exist.

We then use a more general notion of popularity, that allows agents to have different weights. A matching $\mu$ is \textbf{more w-popular} than matching $\mu'$ if the sum of the weights of the agents who strictly prefer their assignments under $\mu$ to those under $\mu'$ is larger than that of the agents who strictly prefer their assignments under $\mu'$ to those  under $\mu$. Group decision involving weighted voting similar to the one implied by this notion of popularity is common. Examples include stockholder voting in corporations, and elections for university deans \citep{Lucas1983}. Moreover, in many real-life environments, the idea that participants do not have the same ability to influence the allocation is natural: senior faculty have a bigger influence than juniors, newcomers have less influence than veterans, etc. Even with this more general model, however, the existence of a w-popular matching is hard to guarantee: unless the weights are \textit{cumulatively ordered}, a very restrictive condition, w-popular matchings might still not exist (Proposition \ref{prop:non-existence of an popular matching}).

Our approach, therefore, will be to require maximum compliance with our objective. We say that a mechanism is w-popular if it produces w-popular allocations \textit{whenever they exist}. The space of w-popular mechanisms includes every possible combination of allocations for problems in which a w-popular matching does not exist. For example, every mechanism that produces a w-popular matching when one exists and satisfies any 
``second-best'' property when one does not exist, is w-popular.

Although the space of w-popular mechanisms is extensive, we obtain a series of well-defined results. We characterize the profiles of weights of the agents that allow for w-popular and strategy-proof mechanisms to exist and show that when they do, serial dictatorships (SDs) in which the order of the agents is consistent with their weights are w-popular and strategy-proof. In addition, the set of outcomes of SDs consistent with the weights fully characterizes the set of w-popular allocations (Theorem \ref{thm:MainResultsBeforeCharacterization}). Since SDs are Pareto efficient, these mechanisms are as well. We show that these profiles of weights also characterize the scenarios in which w-popular matchings can be implemented in Nash equilibrium (Theorem \ref{thm:no-pop-in-equilibrium}).   

When the profile of weights allows for w-popular and strategy-proof mechanisms to exist, we obtain a characterization for SD: a mechanism is strategy-proof, non-wasteful, w-popular and preserves dispute resolutions (a weak consistency requirement) if and only if it is an SD consistent with the weights of the participants (Theorem \ref{thm:SDCharacterization}).

These results provide, therefore, a new case for the use of SDs in problems of allocation of objects with common ownership. Whenever the problem involves weights for the participants, either by its nature or by design, and if these weights allow for a strategy-proof and w-popular mechanism to exist,  then SDs are not only strategy-proof and Pareto efficient but also w-popular. Any other mechanism that is also strategy-proof and w-popular would have undesirable inconsistency or wasteful properties.

\subsection{Related literature}
\label{subsec:RelatedLit}

Our paper is related to, among other subjects, the implementation of allocations immune to coalitional deviations, the properties and existence of popular matchings, weighted voting procedures, and characterizations of serial dictatorships.

In our setup, some groups of agents might challenge an entire allocation with another one. \cite{Moulin1982-qc} propose the notion of \textit{effective coalitions}, a generalization of the concept of core in which coalitions of agents might be able to ``force'' the final decision between subsets of allocations. This concept, however, does not generalize or is generalized by the challenges that we describe. \cite{RubinsteinJungle2007} describe equilibria in an exchange economy in which an individual with more ``power'' than the other can take away her goods. They show that equilibria in allocation problems similar to ours can be constructed via an SD that follows agents' ``power''. Among other differences, their setup does not consider coalitional actions or implementation questions.

The concept of popular matching was introduced by \cite{Gardenfors1975-hs} in the context of a one-to-one marriage market.The author shows that in that context, when both sides have strict preferences, popular matchings exist and are Pareto efficient.

For the one-sided problem of matching agents to objects, \cite{Abraham2007-iq} provide a characterization for the set of popular matchings, and an algorithm for finding one or determining that one does not exist. Other papers extend these results with other algorithms, such as \cite{McDermid2009-gp,Manlove2006-jx}. \cite{kondratev2021minimal} provide a characterization of popular matchings in terms of majority-based 3-way exchanges, and establish a relationship between popularity and minimization of envy. A key concept that we will use is that of \textit{weighted popular matchings}, considered by \cite{mestre_weighted_2014} for the case where objects have unit capacity, and \cite{Sng2010-uk} for when there might be multiple units of the objects. Both provide algorithms for finding weighted popular matchings or determining that they do not exist. \cite{itoh_weighted_2010} establishes values for the probability that a random matching is weighted popular when only two weight values are possible. \cite{cseh2017popular} provides a review on the subject of popular matchings.

When it comes to the evaluation of incentives when implementing popular matchings, the literature is scarce. \cite{Nasre2014-le} evaluates the incentive to misrepresent preferences when facing a mechanism that produces popular matchings. The author finds that there are instances in which agents can improve their outcome by misrepresenting their preferences. The result is, however, limited by the fact that she assumes that a popular matching always exists, regardless of the manipulations made by the agents. \cite{Aziz2013-rl} consider a notion of popularity applied to random assignments, which, as opposed to the deterministic version, has existence guaranteed. The authors show that popular matchings are Pareto efficient and that there are popular random assignments that satisfy notions of equal treatment of equals and envy-freeness. Among other results, they show that when there are more than two agents, there is no ``strategy-proof'' and popular mechanism. Since popularity can always be satisfied in their framework, and their use of randomization and expected utility, these results do not imply ours.

Even though we tackle the problem from a matching perspective, this study is conceptually related to voting.  Specifically, the notion of w-popularity is essentially counterpart of weighted voting in a matching framework. Among many others, \cite{Yoram} consider a weighted voting game where they show that agents can benefit by creating virtual agents and splitting their weights. In \cite{meir}'s strategic weighted voting formulation,  agents repeatedly  change their votes upon observing the others'  votes. They study how convergence to a Nash equilibrium in the game depends on various attributes, including the agents' weights.

Finally, this paper also makes contributions to the works exploring the properties of SD, and its characterizations. For settings, like ours, in which agents may consume at most one unit of an object, \cite{Svensson1994-ur} shows that truth-telling is a dominant strategy in an SD and that every Pareto efficient allocation is the result of an SD for some ordering. \cite{sven} provides a characterization of an SD showing that it is the only strategy-proof, non-bossy and neutral mechanism. \cite{afa} shows that an SD is the only rule that, when used repeatedly, satisfies non-wastefulness, robust strategy-proofness, and some other consistency properties. When considering agents with multi-unit demand, there are multiple results characterizing an SD in these domains, such as \cite{Papai2001-xr}, \cite{hat}, \cite{Ehlers2003-hd}, \cite{Klaus2002-ir}, among others. 

\section{Model}
\label{sec:Model}

Let $N$ and $O$ be the finite sets of agents and objects, respectively. Each agent $i\in N$ has a strict preference ordering $P_{i}$ over $O$ and the outside option of being unassigned, denoted by $\emptyset$. Let $P=(P_i)_{i\in N}$.  Object $c$ is \textbf{acceptable} to agent $i$ if $c \mathrel{P_i} \emptyset$, and otherwise, it is \textbf{unacceptable}. We write $R_i$ for the ``at-least-as-good-as" relation. Let $q=(q_{a})_{a\in O}$ be the capacity profile of the objects. Let $\mathcal{P}$ be the set of all strict preferences.

In the rest of the paper, we fix all the elements except the agent preferences and object quotas and simply write  $(P,q)$ for the problem. We assume that $n=|N|\geq 3$ and $m=|O|\geq n$. All the omitted proofs from the main body are relegated to the Appendix.

A \textbf{matching} $\mu$ is an assignment of objects to agents such that no agent obtains more than one object, and no object is assigned to more agents than its capacity. We write $\mu_k$ for the assignment of agent (or object) $k\in N\cup O$ under the matching $\mu$. A matching $\mu$ is \textbf{non-wasteful} if there is no agent-object pair $(i,a)$ such that  $|\mu_a| < q_a$ and $a \mathrel{P_i} \mu_i$. A matching $\mu'$ \textbf{Pareto improves} upon $\mu$ if for every $i\in N$, $\mu'_i \mathrel{R_i} \mu_i$ and for some $j\in N$, $\mu'_j \mathrel{P_j} \mu_j$. A matching $\mu$ is \textbf{Pareto efficient} if no matching Pareto improves upon it. A \textbf{mechanism} $\psi$ is a function that produces a matching for each problem.  We write $\psi(P,q)$ to denote its outcome in problem $(P,q)$. A mechanism $\psi$ is $[$ \emph{non-wasteful, Pareto efficient} $]$ if for each problem $(P,q)$, $\psi(P,q)$ is $[$ non-wasteful, Pareto efficient $]$. 

\begin{defn}
A matching $\mu$ is \textbf{more popular} than matching $\mu'$ if $\left|\{i\in N:\ \mu_i \mathrel{P_i} \mu'_i\}\right| > \left|\{i\in N:\ \mu'_i \mathrel{P_i} \mu_i\}\right|$. A matching is \textbf{popular} if no matching is more popular than it. A mechanism is \emph{popular} if its outcome is always  popular whenever a popular matching exists.
\end{defn}

When considering this notion of popularity, a popular matching might not exist. Moreover, no mechanism is strategy-proof and  popular (Corollaries \ref{cor:popMatchingMayNotExist} and \ref{cor:noSPPopularMech} below). We will, therefore, relax the definition of popularity in what follows.

\section{Weights and Popularity}
\label{sec:WeightsAndPopularity}

We now assume that there is a \textbf{weight profile} $w$, that associates each agent  $i\in N$ with a positive \textbf{weight} $w_i >0$. Moreover, let $N=\{i_1,..,i_n\}$ be an enumeration of the agents such that $w_{i_k} \geq w_{i_{k'}}$ for each $k < k'$. Throughout the rest of the paper, unless otherwise stated, we fix this ordering and use it whenever mentioned.

\begin{defn}
A matching $\mu$ is \textbf{more w-popular} than matching $\mu'$ if $\sum_{i\in N:\ \mu_{i} \mathrel{P_{i}} \mu'_{i}}\omega_{i}>\sum_{j\in N:\ \mu'_{j} \mathrel{P_{j}} \mu_{j}} w_{j}$. A matching is \textbf{w-popular} if no matching is more w-popular than it. A mechanism is \emph{w-popular} if its outcome is always w-popular whenever a w-popular matching exists. 
\end{defn}

Since Pareto improvements always yield  more w-popular matchings than initial allocations, one fact that comes immediately from the definition of w-popular matchings is given in the following.

\begin{fact}
Every w-popular matching is Pareto efficient. 
\end{fact}

The following condition on weights will have an important role in some of our results. 

\begin{defn}\label{def:sequentiallyCumulativeWeights}
A weight profile $w$ is \textbf{cumulatively ordered} if for each $i_j\in N$, $w_{i_j} \geq \sum_{k > j} w_{i_k}$
\end{defn}

If $w$ is  not cumulatively ordered, then there exists a sequence of the weights $w_{i_k} \geq ...\geq w_{i_{k'}}$ such that $w_{i_k} < \sum_{k < j \leq k'} w_{i_j}$.

\begin{prop}\label{prop:non-existence of an popular matching}
If the weights are cumulatively ordered, then there always exists a w-popular matching. Otherwise, there exists a problem in which no w-popular matching exists.
\end{prop}

Proposition \ref{prop:non-existence of an popular matching} is, essentially, a negative result. It shows that unless the weights are cumulatively ordered, a strong condition, there are problems where a popular matching does not exist. That is, allowing agents to have different weights does not, in general, solve the existence problem of popular matchings. It also constitutes a statement of the \textit{Condorcet paradox} for almost every weight profile: unless they are cumulatively ordered, there are profiles of preferences for which agents cycle through matchings that are more w-popular than the previous.

Notice, moreover, that since equal weights are not cumulatively ordered, we obtain the following corollary:

\begin{cor}\label{cor:popMatchingMayNotExist}
A popular matching may not exist.
\end{cor}

The fact that popular matchings might not exist---and therefore w-popular matchings might also not exist for arbitrarily given weights profiles---was already known, since  Corollary \ref{cor:popMatchingMayNotExist} is a widely known fact in the literature on popular matchings. While the existing literature in weighted popular matchings, listed in section \ref{subsec:RelatedLit}, describes algorithms for finding one or determining it does not exist, none of them established conditions on the weights for which they might or not exist. 

To understand why cumulatively ordered weights are central to the existence of w-popular matchings, consider the example below.

\begin{example}
Let $N=\{i_1,i_2,i_3\}$ and $O=\{o_1,o_2,o_3\}$, and consider two weight profiles: $w=(6,3,2)$ and $w'=(4,3,2)$ (the nth value being the weight of agent $i_n$). The weight profile $w$ is, therefore, cumulatively ordered, while $w'$ is not. Under $w$, it is easy to produce a w-popular matching: $i_1$ must get her most preferred object, otherwise she could challenge that matching, since $6>3+2$. Similarly, $i_2$ must get her most preferred remaining object, and $i_3$ the last remaining object, if that is acceptable to her. By construction, no matching is more w-popular, and therefore this constitutes a w-popular matching.

Consider next the weight profile $w'$, and assume that for all $i\in N$, $o_1 \mathrel{P_i} o_2 \mathrel{P_i} o_3 \mathrel{P_i} \emptyset$, and suppose that $\mu$ is a w-popular matching. Clearly, under $\mu$ every object is matched to an agent, otherwise a matching that matches the unmatched object(s) to the unmatched agent(s) would be more w-popular. Let $i$ and $i'$ be such that $\mu_i=o_1$ and $\mu_{i'}=o_2$, and therefore $\mu_{i''}=o_3$. Regardless of the identities of $i$, $i'$ and $i''$, the matching $\mu'$, in which $\mu'_{i'}=o_1$ and $\mu'_{i''}=o_2$ is more w-popular than $\mu$, contradicting the assumption that $\mu$ was w-popular.
\end{example}

Whenever a weight profile is not cumulatively ordered, there are problems in which three agents might cycle through matchings that are more w-popular, as in the example above, because any coalition of two agents within them can reshuffle the allocation in a way that makes them better off. 

\section{Implementation}

In this section, we evaluate the possibility of implementing w-popular matchings using direct mechanisms. A mechanism $\psi$ is \textbf{strategy-proof} if there are no problem $P$ and agent $i$ with false preferences $P'_i$ such that $\psi_i(P'_{i}, P_{-i},q) \mathrel{P_i} \psi_i(P,q)$.\footnote{$P_{-i}$ is the preference profile of all the agents but agent $i$.}  For a given ordering of agents $(i_1,i_2,\ldots,i_n)$, \textbf{SD (serial dictatorship) mechanism} follows the ordering of agents, giving each one the most preferred acceptable object that was not yet assigned previously. We say that a $SD$ is \textbf{consistent with the agent's weight profile} if for any pair of agents $i,j\in N$ with $w_i>w_j$,  agent $i$ comes before agent $j$ in the ordering that is used in the $SD$. Notice that whenever agents have the same weight, then any ordering among themselves complies with consistency.
We say that a weight profile is \textbf{distinct} if no weight is the same as some other. A weight profile $w$ is \textbf{essentially distinct} if for all $j<k\leq n-1$, $w_{i_j}\neq w_{i_k}$, but $w_{i_{n-1}}=w_{i_n}$ and $w_{i_{n-2}} \geq w_{i_{n-1}}+w_{i_n}$. The following fact directly comes from the definitions.

\begin{fact}
Cumulatively ordered weight profiles are either distinct or essentially distinct.
\end{fact}

The theorem below establishes the space of weight profiles that allow for strategy-proof w-popular mechanisms, and the intimate relationship between w-popularity and SDs.

\begin{thm}\label{thm:MainResultsBeforeCharacterization}

\leavevmode
\makeatletter
\makeatother
\begin{itemize}
\item[(i)] A $w$-popular and strategy-proof mechanism exists if and only if the weight profiles are either distinct or essentially distinct.
\item[(ii)] If a weight profile is distinct or essentially distinct, then the $SD$s that are consistent with the weight profile are  $w$-popular and strategy-proof.
\item[(iii)] At each problem where a $w$-popular matching exists, any $w$-popular matching is the outcome of a $SD$  that is consistent with the weight profile. 
\end{itemize}
\end{thm}

Notice that, given the fact above, items (i) and (ii) also hold for cumulatively ordered weight profiles. Notice, moreover, that item (iii) also implies that there is at most one w-popular matching when weight profiles are distinct, and at most two w-popular matchings when they are essentially distinct.  When considering the standard notion of popularity, Theorem \ref{thm:MainResultsBeforeCharacterization} implies the following:

\begin{cor}\label{cor:noSPPopularMech}
No mechanism is popular and strategy-proof.
\end{cor}

While Theorem $1$ shows that $SD$ is $w$-popular and strategy-proof, there might be other such rules. To pin down $SD$ in this class of mechanism, we provide a new characterization of $SD$, based on $w$-popularity and strategy-proofness. Below first introduces the characterization axioms.

For an object $a$, let $\mathcal{P}_a$ be the class of preferences where the only acceptable object is object $a$. Likewise, we write $\mathcal{P}_{\emptyset}$ for the class of preferences where no object is acceptable. 

A mechanism $\psi$  \textbf{preserves dispute resolutions}  whenever for a problem $(P,q)$, a pair of agents $i,j$,  $P'_i\in \mathcal{P}_a$,  and object $a$ such that $a=\psi_j(P,q) \mathrel{P_i} \psi_i(P,q)$ and $\psi_i(P'_i, P_{-i}, q)=\emptyset$,  $\psi_i(P'', q')=\emptyset$ and $\psi_j(P'', q')=a$ for each $P''$ and $q'$ where $q'_a=1$, $P''_i, P''_j\in \mathcal{P}_a$, and $P''_k\in \mathcal{P}_{\emptyset}$ for each other agent $k$.\footnote{\cite{bettina} characterize the class of deferred-acceptance mechanisms under responsive priorities. One of their axioms is ``two-agent consistent conflict resolution." While it shares some similarity with preserving dispute resolution, these two axioms are independent. Formally speaking, a mechanism satisfies two-agent consistent conflict resolution if, for each pair of agents $i,j$, object $a$,  capacities $q$, $q'$, and preference profile $P$, where $P_i, P_j\in \mathcal{P}_a$, $\{\psi_i(P,q), \psi_j(P, q')\}=\{\psi_i(P, q'), \psi_j(P, q')\}=\{a, \emptyset\}$, then for each $k\in \{i,j\}$, $\psi_k(P,q)=\psi_k(P,q')$. As preserving dispute resolutions has a say only when all the agents except $i,j$ find no object acceptable, it does not imply two-agent consistent conflict resolution. Similarly, two-agent consistent conflict resolution kicks in whenever both agents $i,j$ have the same and unique acceptable object, hence it does not   imply preserving dispute resolutions either.}

In words, if an agent $i$ envies agent $j$ for her object $a$ and the former continues not receiving object $a$ even after she declares it as her only acceptable object, then a mechanism that preserves dispute resolutions will still make that allocation of $a$ when the problem is reduced to the allocation of a single copy of $a$ between $i$ and $j$. It is, therefore, a relatively weak consistency property for deterministic mechanisms, satisfied, for example, by mechanisms that produce stable allocations \citep{Gale1962-mn}.

\begin{thm}\label{thm:SDCharacterization}
Suppose that the weight profile is distinct or essentially distinct. Then,  a mechanism is non-wasteful, w-popular, strategy-proof, and preserves dispute resolutions if and only if it is a $SD$ that is consistent with the weight profile. 
\end{thm}

Note that here again the statement also holds for cumulatively ordered weight profiles. We show the independence of the axioms in the Appendix.  Theorem \ref{thm:MainResultsBeforeCharacterization} tells us that whenever the weight profile is  distinct or essentially distinct, we can implement w-popular matchings with a strategy-proof mechanism. Below, we find that outside of this weight domain, Nash implementation is impossible as well. Each direct mechanism induces a preference reporting game among agents.  Let $\Omega$ be the set of mechanisms admitting a (pure) Nash equilibrium. For $\psi\in \Omega$, we say that $\psi$ is \textbf{w-popular in equilibrium} if for each problem $(P,q)$ and each Nash equilibrium $P'$ at $(P,q)$, $\psi(P',q)$ is w-popular whenever a w-popular matching exists. 

\begin{thm}\label{thm:no-pop-in-equilibrium}
Let $w$ be a weight profile. $\Omega$ contains a mechanism that is w-popular in equilibrium if and only if the weight profile $w$  is  distinct or essentially distinct.
\end{thm}

As cumulatively ordered weights are either distinct or essentially distinct, the above result holds under them as well. 

\begin{rem}
\label{rem:SDEquilibria}
Note that $SD$ admits weakly dominated Nash equilibria. However, it is immediate to see that every Nash equilibrium yields the same outcome. This, along with Theorem \ref{thm:MainResultsBeforeCharacterization}, implies that  $SD$s that are consistent with the weight profile are $w$-popular in equilibrium whenever the weight profile is  distinct or essentially distinct.
\end{rem}

\bibliographystyle{ecta}
\bibliography{popMechanisms}

\section*{Appendix}

\begin{proof}[Proof of Proposition \ref{prop:non-existence of an popular matching}]
Let us first assume that the weight profile is cumulatively ordered. Given a problem $(P,q)$,  let us produce a matching by letting each agent choose her best remaining object, one by one, following the agent-ordering.\footnote{In other words, we obtain a matching by running a serial dictatorship.} Recall that we have the ordering of $N=\{i_1,..i_n\}$ where $w_{i_k} \geq w_{i_{k'}}$ for each $k < k'$.\footnote{There can be multiple orderings in decreasing order of the weights. We just arbitrarily fix one of them, implying that all of our results hold for any such ordering.} Let us denote the matching obtained by $\mu$.  By definition, agent $i_1$ is matched with her top choice, then agent $i_2$ is matched with her top choice among the remaining ones, and so on.

We now show that $\mu$ is w-popular. Assume for a contradiction that $\mu'$ is more w-popular than $\mu$.  Let $i_j$ be the earliest coming agent in the ordering such that $\mu_{i_j} \neq \mu'_{i_j}$. As agent $i_j$ obtains the best remaining object after excluding the assignments of those coming earlier in the ordering and the latter group have the same assignment under both $\mu$ and $\mu'$,  we have $\mu_{i_j} \mathrel{P_j} \mu'_{i_j}$.  But then, because $w_{i_j} \geq \sum_{k > j} w_{i_k}$, $\mu'$ cannot be more popular than $\mu$ even if each agent who comes later than $i_j$ in the ordering obtains a better assignment under the former. This shows that $\mu$ is w-popular.

For the other part, consider a problem where each object has the quota of one and each agent $i\in N$ has the same preferences: $a_1 \mathrel{P_i} a_2 \mathrel{P_i} \cdots \mathrel{P_i} a_m \mathrel{P_i} \emptyset$.  Let $w_{i_\ell}$ be the lowest value of $\ell$ for which the cumulative ordered condition fails. That is, $w_{i_\ell} < \sum_{k > \ell} w_{i_k}$.  By construction, $w_{i_1} >...> w_{i_{\ell-1}} > w_{i_{\ell}}$.

For contradiction, let $\mu$ be a w-popular matching. As $w_{i_1} >...> w_{i_{\ell-1}}$, and each of these weights is strictly larger than 
all the other weights, under any popular matching, hence in particular under $\mu$,  $\mu_{i_1}=a_1$, $\mu_{i_2}=a_2$, etc, and $\mu_{i_{\ell-1}}=a_{\ell-1}$. Because otherwise either by swapping these agents' assignments with the assignments of the agents with lower weights or giving their more preferred leftover objects, we can obtain a more w-popular matching than $\mu$.  Moreover,  because of the w-popularity of $\mu$, all the other agents receive an object from $\{a_{\ell},..,a_n\}$. 

Let $\mu_{i_{k^*}}=a_{\ell}$ for some $k^*\in \{\ell,..,n\}$.  Then, $\mu$ matches $n-\ell$ objects in $\{a_{\ell+1},..,a_n\}$ to the agents in $\{i_{\ell},i_{\ell+1},\ldots,i_{n}\}\setminus \{i_{k^*}\}$. Let us consider an alternative matching $\mu'$ such that $\mu_{i_k}=\mu'_{i_k}$ for each $k < \ell$.  For each $k,k'\in \{\ell,..,n\}\setminus \{k^*\}$, $\mu'_{i_k}=a_{k'-1}$  where $\mu_{i_k}=a_{k'}$. Lastly, let $\mu'_{i_{k^*}}=a_n$. By definition, all the agents in $\{i_{\ell},..,i_n\}\setminus \{i_{k^*}\}$ strictly prefer their assignments under $\mu'$ to those under $\mu$. All the others are indifferent, except $i_{k^*}$, who prefers her assignment under $\mu$ to $\mu'$. Given all of these as well as  $w_{i_\ell} < \sum_{k > \ell} w_{i_k}$ and $w_{i_{\ell}} \geq w_{i_{\ell+1}}\geq ...\geq w_{i_n}$,  $\mu'$ is more w-popular than $\mu$, contradicting our starting supposition.
\end{proof}

\begin{proof}[Proof of Theorem \ref{thm:MainResultsBeforeCharacterization}]

We start with items (ii) and (iii). Since the weight profile is either distinct or essentially distinct, we have $w_{i_1} > w_{i_2} > ...... > w_{i_{n-1}} \geq w_{i_n}$. That is, only the last two agents may have the same weight.  Let $\mu$ be a popular matching. Then, under $\mu$, each agent $i_k\in \{i_1,...,i_{n-2}\}$ receives her top remaining choice one by one following the agent ordering. To see this, suppose not, and let $i$ be the first agent in this ordering who does not, and let $a$ be the object that is her top remaining choice following the agent ordering. Then either $a$ is not assigned to anyone---in which case a matching that is equal to $\mu$ but assigns $a$ to $i$ is more popular than $\mu$---or $a$ is matched to an agent $j$ below $i$ in the ordering, implying $w_{j}<w_{i}$---in which case a matching that is equal to $\mu$ except that is switches the assignments of $i$ and $j$ is more popular than $\mu$. Both cases would contradict the popularity of $\mu$.  

Let us now consider $i_{n-1}$ and $i_n$. If $w_{i_{n-1}} > w_{i_n}$ (the distinct weight profile case), the same is true for these two agents. That is, each obtains her top remaining choice by following the weight ordering, showing that the outcome of an $SD$ that is consistent with the weight profile is the unique $w$-popular matching. Therefore, when weights are distinct and a $w$-popular matching exists, it is precisely the outcome of this $SD$. Otherwise, we have $w_{i_{n-1}}=w_{i_n}$. That is, $w$ is essentially distinct, hence $w_{i_{n-2}} \geq w_{i_{n-1}}+w_{i_{n}}$.  In this case, these agents alternatively receiving their top remaining choice across problems is compatible with $w$-popularity. This, in turn, shows that  each $SD$ that is consistent with the weight profile is $w$-popular. We already know that each $SD$ is strategy-proof, finishing the proof of item (ii).

Let us consider a problem where a $w$-popular matching exists. Let $\mu$ be a $w$-popular matching, and assume that it is not the outcome of a $SD$ that is consistent with the weight profile. This implies that for some agent $i$ and object $a$, $a \mathrel{P_i} \mu_i$, and  either $\mu_j=a$ for some agent $j$ with $w_j < w_i$ or $|\mu_a| < q_a$. For the latter (case), we let agent $i$ receive object $a$, while keeping the others' assignments the same. For the former (case), we let agents $i$ and $j$ trade their assignments, while keeping the others' assignments the same. In both cases, we obtain a matching that is more $w$-popular than $\mu$, contradicting our starting supposition, finishing the proof of item (iii).

Next, we proceed with item (i). Let $w$ be a weight profile that does not satisfy either of the properties. We then have two cases to consider.

\textbf{Case 1.} $w_{i_{n-2}} < w_{i_{n-1}}+w_{i_n}$ and $w_{i_{n-1}}=w_{i_n}$. Let us consider a problem where $\{a,b,c\}\subseteq O$ and $q_a=q_b=q_c=1$. Let the preferences be such that $P_{i_{n-2}}=P_{i_{n-1}}=P_{i_n}:\ a, b, c, \emptyset$; and all the other agents (if any) find each object unacceptable. Let $\psi$  be a w-popular mechanism and $\psi(P,q)=\mu$. 

Since there are only three agents with acceptable objects, with the same preferences, and none of them have weights strictly larger than the sum of the other two,  there is no w-popular matching in problem $P$. While a $w$-popular mechanism could produce any matching in this scenario, it will in any case be true that for some agent $j\in \{i_{n-2}, i_{n-1}, i_n\}$, $b$ is preferred to her assignment in that matching, that is, $b \mathrel{P_j} \mu_j$. Let us consider  $P'_{j}: b, \emptyset$, and write $P'=(P'_{j}, P_{-j})$. At problem $P'$, there are two w-popular matchings: $\mu'$ and $\mu''$, where under both matchings, agent $j$ is assigned to the object $b$.\footnote{At $\mu'$, either of the agents in $\{i_{n-2}, i_{n-1}, i_n\}\setminus \{j\}$ receives object $a$, and the other receives object $c$. The converse occurs at $\mu''$.} Hence, either $\psi(P',q)=\mu'$ or $\psi(P',q)=\mu''$, showing that agent $j$ benefits from the misreporting.

\textbf{Case 2.} We have $w_{i_{k}}=w_{i_{k+1}}$ where $k+1 < n$. Then, let us consider a preference profile $P$ where $P_{i_k}=P_{i_{k+1}}=P_{i_{k+2}}: a, b, c, \emptyset$; and each other agent (if any) finds each object unacceptable. Note that since $w_{i_k} < w_{i_{k+1}}+w_{i_{k+2}}$, and these are the only agents with acceptable objects, once again there is no popular matching. Moreover, using the same reasoning as in Case 1, we get that $\psi$ cannot be strategy-proof.

\end{proof}

\begin{proof}[Proof of Theorem \ref{thm:SDCharacterization}]
Recall that $N=\{i_1,..,i_n\}$ where $w_{i_k} \geq w_{i_{k'}}$ for each  $k \geq k'$. Let us suppose that $w$ is distinct or essentially distinct. From Theorem $1$, we know that $SD$s that are consistent with the weight profile are both w-popular and strategy-proof. They are also non-wasteful. One can also easily verify that they preserve dispute resolutions as well.

Let $\psi$ be a non-wasteful, w-popular,  and strategy-proof mechanism that also preserves dispute resolution.
Let $(P,q)$ be a problem where a w-popular matching exists. Let $\psi(P,q)=\mu$. In light of Theorem \ref{thm:MainResultsBeforeCharacterization}, there are two classes of weight profiles to consider.

\textbf{Case 1.} Suppose $w$ is distinct. Each agent $i_k$ has to receive her remaining top choice one by one following the weight ordering under $\mu$. Otherwise, let $i_k$ be the agent who comes earliest among the agents whose $SD$ assignments are different than those under $\mu$ (note that as $w$  is distinct, there is only one $SD$ that is consistent with the weight profile). Let $\mu_{i_k}=b$ while her $SD$ assignment is object $a$. By our construction,  each agent $i_{k'}$ with $k' < k$ receives the same object under both $\mu$ and the $SD$ outcome. By the definition of $SD$, it implies that $a \mathrel{P_{i_k}} b$. By the non-wastefulness of $\mu$, there exists an agent $i_{k'}$ with $k' > k$, hence $w_{i_{k'}} < w_{i_k}$, such that $\mu_{i_{k'}}=a$. This matching $\mu$, however, cannot be w-popular as one can obtain a more w-popular matching by swapping the assignments of the agents $i_k$ and $i_{k'}$ under $\mu$. Hence, $\mu$ is nothing but the $SD$ outcome that is consistent with the weights.

\textbf{Case 2.} Suppose that $w$ is not distinct, but essentially distinct. We have  $w_{i_1} > w_{i_2} ...> w_{i_{n-1}}=w_{i_{n}}$.  For each $k \leq n-2$, agent $i_k$ has to receive her remaining top choice, one by one, following the weight ordering under $\mu$. The same arguments as above easily shows that otherwise we could obtain a more popular matching. 
By the non-wastefulness of $\mu$, the rest, i.e., agents $i_{n-1}$ and $i_n$, receive their best remaining choice one by one in any order. This, however, means that $\mu$ is nothing but the outcome of the either of $SD$s that are consistent with the weights.

Let $(P,q)$ be a problem in which a w-popular matching does not exist. Let $\psi(P,q)=\mu'$. If $\mu'$ is  the outcome of the either of $SD$s that are consistent with the weight profile (there might be one or two of them), then there is nothing to prove. Otherwise, let $\mu$ be the outcome of either of these $SDs$. We have $\mu'\neq \mu$. Let $i_k$ be the first agent in the ordering such that $\mu_{i_k}\neq \mu'_{i_k}$. Let $\mu_{i_k}=a$. By the definition of $SD$, object $a$ is the top object for agent $i_k$ after excluding the assignments of the agents  who come earlier than herself in the $SD$ ordering, hence $a \mathrel{P_{i_k}} \mu'_{i_k}$.  This, as well as the non-wastefulness of $\mu'$, implies that for some agent $i_{k'}$,  $\mu'_{i_{k'}}=a$. 

Moreover, as for each agent who comes earlier than $i_k$  in the SD ordering, her assignments  under both $\mu$ and $\mu'$ are the same, we have $w_{i_k} \geq w_{i_{k'}}$. 
We now claim that the relation is strict. Assume for a contradiction that $w_{i_k}=w_{i_{k'}}$. This means that the weight profile is not distinct, hence by our supposition, it is essentially distinct.  Moreover, the agents with these two weights are the last two at any agent-ordering that is consistent with the weight profile. From  above, for each agent $j$ with $w_j > w_{i_k}$, we have $\mu_j=\mu'_j$ (note that the weights are strictly ordered until the last two). But then, the non-wastefulness of $\mu'$ implies that the last two agents, $i_k$ and $i_{k'}$, receive their remaining top choice one by one in any order among themselves. This in turn implies that $\mu'$ is nothing but the outcome of the either of $SD$s that are consistent with the weight profile, yielding a contradiction. Hence, $w_{i_k} > w_{i_{k'}}$. 

Let us now go back problem $(P,q)$ where  $a \mathrel{P_{i_k}} \mu'_{i_k}$ and $\mu'_{i_{k'}}=a$. From above, we also know that $w_{i_k} > w_{i_{k'}}$. By the strategy-proofness of $\psi$,  for any $P'_{i_k} \in \mathcal{P}_a$,  $\psi_i(P'_{i_k}, P_{-i_k}, q)=\emptyset$. Therefore, since $\psi$ preserves dispute resolutions, for each $(P'',q')$ where $q'_a=1$, both $P''_{i_k}, P''_{i_{k'}}\in \mathcal{P}_a$ and $P''_k\in \mathcal{P}_{\emptyset}$ for each other agent $k$,   we have $\psi_{i_k}(P'',q')=\emptyset$ and $\psi_{i_{k'}}(P'',q')=a$. This matching, however, is not w-popular as the unique w-popular matching $\mu''$ at $(P'',q')$ is such that $\mu''_{i_k}=a$ and $\mu''_j=\emptyset$ for each other agent $j$. This contradicts the w-popularity of $\psi$, finishing the proof.
\end{proof}

\emph{Independence of the Axioms in Theorem \ref{thm:SDCharacterization}}

\textbf{A mechanism that satisfies all the axioms except w-popularity}: It is easy to see that \cite{Gale1962-mn}'s deferred-acceptance mechanism satisfies all the properties except w-popularity.

\textbf{A mechanism that satisfies all the axioms except preserving dispute resolutions}: Let us now consider a problem where $N=\{i_1,..,i_5, i_6\}$ and $O=\{a_1,..,a_5, a_6\}$. Let $q$ be such that $q_a=1$ for each object. Let the weights be $20, 10, 5, 4, 3, 2$. Let $P_{i_1}=P_{i_2}:\ a_1, a_2, \emptyset$; and $P_k:\ a_1, a_2, a_3, a_4, a_5, a_6, \emptyset$ for each other agent $k$.  If we write $N'=\{i_1, i_2\}$, then  no popular matching exists in $(P'_{N'}, P_{-N'},q)$ for any $P'_{N'}\in \mathcal{P}^{|N'|}$.\footnote{For any $N'\subset N$, $P_{N'}$ stands for the preference profile of the agents in $N'$.} Let $\psi$ be the mechanism such that at any problem $(P'_{N'}, P_{-N},q)$ where $P'_{N'}\in \mathcal{P}^{|N'|}$, it  gives the $SD$ outcome under the ordering of $i_2, i_1,i_3, i_4, i_5, i_6$.  Otherwise, it gives the $SD$ outcome where the ordering is consistent with the weights. This mechanism $\psi$ satisfies all the properties except preserving dispute resolutions (because agent $i_1$ receives object $a_1$ whenever she and $i_2$ report only object $a_1$ acceptable while others report no object acceptable).

\textbf{A mechanism that satisfies all the axioms except non-wastefulness}: Let us consider a  problem $N=\{i_1,...,i_4\}$ and $O=\{a_1,..,a_4\}$.  Let $q$ be such that $q_a=1$ for each object. Let the weights be $7, 5, 3, 1$. Let $P_{i_1}=P_{i_2}=P_{i_3}:\ a_1, a_2, a_3, \emptyset$; and $P_{i_4}:\ a_4, \emptyset$. Let $\psi(P,q)=\mu$ where $\mu_{i_k}=a_k$ for each $k \leq 3$ and $\mu_{i_4}=\emptyset$. Note that if no popular matching exists in a problem, then it continues not to exist for any preferences of agent $i_4$ keeping the others preferences' the same. Let $\psi$ be such that whenever a w-popular matching does not exists at a problem,  each agent except $i_4$ receives her $SD$ (that is consistent with the weights) assignment while agent $i_4$ is left unassigned. Otherwise, $\psi$ produces the outcome of $SD$ that is consistent with the weights. One can easily verify that $\psi$ satisfies all the properties except non-wastefulness. 

\textbf{A mechanism that satisfies all the axioms except strategy-proofness}: Let us consider $N=\{i,j,k\}$ and $O=\{a,b,c\}$.  Let $q$ be such that $q_a=1$ for each object. Let the weights be $4, 3, 2$. Let $P_i=P_j=P_k:\ a_1 ,a_2 ,a_3, \emptyset$. Note that there is no w-popular matching at problem $(P,q)$. Let $\psi(P,q)=\mu$ where $\mu_{i_1}=a_2$, $\mu_{i_2}=a_1$, and $\mu_{i_3}=a_3$. For all other problems, it is $SD$ that is consistent with the weights. It satisfies all the properties except strategy-proofness.

\begin{proof}[Proof of Theorem \ref{thm:no-pop-in-equilibrium}]
``If" part directly comes from Theorem \ref{thm:MainResultsBeforeCharacterization} and Remark \ref{rem:SDEquilibria}. For the other part, let us consider a weight profile $w$ that does not satisfy any of the properties listed.  Let $\{i_k, i_{k+1}, i_{k+2}\}\subseteq N$ and $\{a_1, a_2, a_3\}\subseteq O$. Let $q_a=1$ for each object $a\in \{a_1, a_2, a_3\}$.  As $w$ satisfies none of the properties, we have the following two cases.

\textbf{Case 1}: $w_{i_k}=w_{i_{k+1}} \geq w_{i_{k+2}}$ where $k+2 < n$ or \textbf{Case 2}: $w_{i_{k+1}}=w_{i_{k+2}}$, $w_{i_k} < w_{i_{k+1}}+w_{i_{k+2}}$, and $k+2=n$ (note that $n$ is the last agent).  The arguments below work for  both cases. For ease of writing, let $i_1=i_k$, $i_2=i_{k+1}$, and $i_3=i_{k+2}$. 

Let $\psi\in \Omega$ be a mechanism that is w-popular in equilibrium. Let us consider a problem $(P,q)$ where  $P_{i_1}=P_{i_2}=P_{i_3}:\ a_1, a_2, a_3, \emptyset$; and each other agent (if any) finds every object unacceptable.  As shown in the proof of Theorem \ref{thm:MainResultsBeforeCharacterization} above, in problem $(P,q)$, there is no w-popular matching. Let $P'$ be an equilibrium in $(P,q)$ and $\psi(P',q)=\mu$. Note that, regardless of what is the matching $\mu$ produced by $\psi$ in this scenario without a popular matching, it must be the case that for some agent $i\in \{i_1, i_2, i_3\}$, $a_2 \mathrel{P_i} \mu_i$. Let $a_2 \mathrel{P_j} \mu_j$, where $j\in \{i_1, i_2, i_3\}$. 

Let us next consider a problem $(P'',q)$ where $P''_{j}:\ a_2, a_3, \emptyset$; and each other agent's preferences is the same as that under $P$.  Under both cases 1 and 2 above, there are only two $w$-popular matchings $\mu$ and $\mu'$, where $\mu_j=\mu'_j=a_2$,  and the agents in $\{i_1, i_2, i_3\}\setminus \{j\}$ alternatively receive $a_1$ and $a_3$ at these matchings.\footnote{Note that if $j=i_2 (i_3)$ and Case $2$ holds, then the unique $w$-popular matching among those is such that agents $i_1$, $i_2$, and $i_3$ respectively receives $a_1$, $a_2 (a_3)$, and $a_3 (a_2)$. Otherwise, both $\mu$ and $\mu'$ are the only $w$-popular matchings.} The remaining players are left unassigned under both matchings, since otherwise, a matching that leaves them unassigned would be more $w$-popular matching than $\mu$ or $\mu'$.

We next claim that $P'$ is an equilibrium at $(P'',q)$ under $\psi$. We first claim that  no agent  $k\in N\setminus \{i_1, i_2, i_3\}$ has incentive to deviate from $P'$.  Assume for a contradiction that for some agent $k\in N\setminus \{i_1, i_2, i_3\}$, we have $\hat{P}_k$ such that $\psi_k(\hat{P}_k, P'_{-k},q) \mathrel{P''_k} \psi_k(P',q)$. As $P''_k=P_k$, it means that $\psi_k(\hat{P}_k, P'_{-k},q) \mathrel{P_k} \psi_k(P',q)$, contradicting $P'$ being an equilibrium at $(P,q)$.
 
We have established that $j$ is such that $a_2 \mathrel{P_j} \psi_j(P',q)=\mu_j$. This means that $\mu_j\in \{a_3, \emptyset\}$. If $\mu_j=a_3$, then the fact that $P'$ is an equilibrium at $(P,q)$ implies that there is no $\bar{P}\in \mathcal{P}$ such that $\psi_j(\bar{P}, P'_{-j},q) \in \{a_1, a_2\}$. Similarly, if $\mu_j=\emptyset$, there is no $\hat{P}\in \mathcal{P}$ such that $\psi_j(\hat{P}, P'_{-j},q) \in \{a_1, a_2, a_3\}$. Note that since $P''$ differs from $P'$ only in agent $j$'s preferences, $\psi(\bar{P}, P'_{-j},q)=\psi(\bar{P}, P''_{-j},q)$ and $\psi(\hat{P}, P'_{-j},q)=\psi(\hat{P}, P''_{-j},q)$. All these show that agent $j$ does not have incentive to deviate from $P'$ at $(P'',q)$.

Let $k\in \{i_1, i_2, i_3\}\setminus \{j\}$. As $P''_k=P_k$ and $P'$ is an equilibrium at $(P,q)$, he does not incentive to deviate from $P'_k$ at  $(P'',q)$. All these show the claim. Thus, under $\psi$, $P'$ is an equilibrium at $(P'',q)$. However, $\psi(P',q)=\mu$ is not $w$-popular at $(P'',q)$. This is because $\mu_j\neq a_2$ whereas at any $w$-popular matching at $(P'',q)$, agent $j$ receives $a_2$.  This contradicts $\psi$ being $w$-popular in equilibrium.

\end{proof}

\end{document}